\documentclass[runningheads]{llncs}
\usepackage[utf8]{inputenc}

\usepackage[utf8]{inputenc}
\usepackage{amssymb}
\usepackage{amsmath}
\usepackage{amsfonts}
\usepackage{makecell}
\usepackage{multirow}
\usepackage{threeparttable}
\usepackage{balance}
\usepackage{flushend}
\usepackage{changepage}
\usepackage{tikz}
\usepackage{graphicx}
\usepackage{diagbox}
\usepackage{array}
\usepackage{comment}
\usepackage{caption}
\usepackage{tabu}
\usepackage{float}

\usepackage{mathrsfs}
\usepackage{pifont}
\usepackage{subfigure}
\usepackage{algorithm}  
\usepackage{algorithmicx}  
\usepackage{algpseudocode}  
\usepackage{multirow} 
\usepackage{amsmath} 
\usepackage{xcolor}
\usepackage{graphicx}
\usepackage{enumerate}

\usetikzlibrary{shadows,arrows}

\begin{document}
\title{CFL: Cluster Federated Learning in Large-scale Peer-to-Peer Networks}
%
%
\author{Qian~Chen\inst{1} \and
		Zilong~Wang\inst{1} \and
		Yilin~Zhou\inst{1} \and
		Jiawei~Chen\inst{1} \and
		Dan~Xiao\inst{1} \and
	    Xiaodong~Lin\inst{2}}
\authorrunning{Q.~Chen et al.}
%
\institute{$^{1}$State Key Laboratory of Integrated Service Networks, School of Cyber Engineering, Xidian University, Xi'an 710071, China \\
\email{\{qchen\_4, jiaweichen98\}@stu.xidian.edu.cn, zlwang@xidian.edu.cn,\\ \{xduyilinzhou, xdudanxiao\}@gmail.com}\\
$^{2}$School of Computer Science, University of
Guelph, Guelph, Canada\\
%
%
%
\email{xlin08@uoguelph.ca}}
\maketitle              
\begin{abstract}
Federated learning (FL) has sparked extensive interest in exploiting the private data on clients' local devices. However, the parameter server setting of FL not only has high bandwidth requirements, but also poses data privacy issues and a single point of failure. In this paper, we propose an efficient and privacy-preserving protocol, dubbed CFL, which is the first fine-grained global model training for FL in large-scale peer-to-peer (P2P) networks. Unlike previous FL in P2P networks, CFL aggregates local model update parameters hierarchically, which improves the communication efficiency facing large amounts of clients. Also, the aggregation in CFL is performed in a secure manner by introducing the authenticated encryption scheme, whose key is established through a random pairwise key scheme enhanced by a proposed voting-based key revocation mechanism. Rigorous analyses show that CFL guarantees the privacy and data integrity and authenticity of local model update parameters under two widespread threat models. More importantly, the proposed key revocation mechanism can effectively resist hijack attacks, thereby ensuring the confidentiality of the communication keys. Ingenious experiments on the Trec06p and Trec07 datasets show that the global model trained by CFL has good classification accuracy, model generalization, and rapid convergence rate, and the dropout-robustness of the system is achieved. Compared to the first global model training protocol for FL in P2P networks, PPT, CFL improves communication efficiency by 43.25\%. Also, CFL outperforms PPT in terms of computational efficiency.

\end{abstract}
\section{Introduction}
Machine learning (ML) has injected vitality into many aspects of modern society, such as speech recognition \cite{speech2013}, image recognition \cite{image2012}, natural language processing \cite{Jordan2015}, etc. Traditional ML is featured by centrally collecting raw data to train the model. However, gathering all data to a central database consumes high communication bandwidth, while the bandwidth is a rare resource. Besides, the centrally collected data has a high risk to be abused, which has drawn more and more attention from society, e.g., the General Data Protection Regulation \cite{GDPR2017} where the usage of user data is strictly limited. Under such circumstances, ML seems to meet a severe data usage dilemma \cite{MLonBD17}.

To address communication efficiency and privacy concerns that centralized ML poses to data owners, Federated learning (FL) \cite{McMahan2017} has emerged as a state-of-the-art ML system to solve the dilemma, which allows clients to collaboratively reap the benefits of shared models without centrally storing data. In FL, clients train their personalized models locally, and then a central server aggregates local contributions, e.g., local models, local gradients, or local model update parameters, to update the global model. Such a distributed paradigm not only preserves the privacy of clients \cite{Bonawitz2017} but also improves the communication efficiency of the system \cite{Konecny2016}, as only the local contributions leave clients' local devices instead of the private and massive local data.

Although FL has been a great success, the parameter server setting relies too much on the central server which might not exist in the real world, such as Internet of Things (IoT) \cite{iot2015}, Smart Home \cite{home2017}, and Ad Hoc \cite{hoc2002}. Besides, the central server gives rise to several drawbacks: (1) the central server must be honest, which is difficult to guarantee in the real world \cite{He-Tan2020}; (2) the central server has high computational costs and high bandwidth requirements \cite{Lian2017}; (3) the central server can become a single point of failure \cite{ring2017,Hu2019}. As a result, how to deploy FL without the central server deserves deep research, which is referred to as the {\it decentralized FL} \cite{Hu2019} or {\it FL in peer-to-peer (P2P) networks} \cite{PPT2021}.

Consequently, some {\it decentralized FL} works involving the protocol \cite{PPT2021}, algorithms \cite{SGP2019,ring2017,Air2019}, and frameworks \cite{Dubey2020,Hu2019,Li-Wen2020,Ramanan2020} have been proposed successively. Also, decentralized vertical FL, which is a special decentralized FL paradigm, has been studied in \cite{He-Tan2020,Marfoq2020,Yang2019}. However, all existing decentralized FL works are coarse-grained, which cannot guide the global model training in practice, except \cite{PPT2021} (Under review by Computers \& Security after a major reversion). Recently, Chen {\it et al.} \cite{PPT2021} proposed a fine-grained PPT protocol, which was claimed as the first communication-efficient and privacy-preserving global model training protocol for {\it FL in P2P networks}. However, PPT has several shortcomings facing massive clients. {\bf First}, PPT aggregates local model update parameters directly by a so-called \textit{cyclic transmission} manner, which is inefficient for {\it FL in large-scale P2P networks} where massive clients are involved. {\bf Second}, the key distribution scheme in PPT is vulnerable to hijacking attacks \cite{hu2019session}, which leaves a potential privacy leakage risk to the system. {\bf Third}, the signature scheme in PPT requires extra storage space and heavy computational power. Therefore, an intuitive question is \textit{How to efficiently and securely achieves FL in large-scale P2P networks.} 

Our response to this question is a \underline{C}luster \underline{F}ederated \underline{L}earning (CFL) global model training protocol for {\it FL in large-scale P2P networks}. Specifically, we expand the system model of PPT \cite{PPT2021}, where clients and a server loosely connected to a few of them are distributed in a P2P network, by involving large amounts of clients. To improve the communication efficiency of the system, CFL aggregates local model update parameters hierarchically, rather than directly aggregates all local contributions like PPT. Instead of leveraging the digital signature scheme, CFL guarantees security and computation efficiency through the authenticated encryption scheme, whose key is established by a random pairwise keys scheme enhanced by a proposed key revocation mechanism which improves the security against hijacking attacks \cite{hu2019session}. 

The main contributions are as follows:

\noindent{\bfseries A global model training protocol for FL in large-scale P2P networks.}
We propose a fine-grained global model training protocol for FL in large-scale P2P networks which comprises a Secure Communication Key Establishment Protocol and an Inner-Cluster Model Aggregation Protocol. Compared to the PPT protocol, larger amounts of clients could train the global model in a higher communication and computation efficient manner with the deployment of the proposed CFL protocol. 

\noindent{\bfseries A secure and privacy-preserving protocol.}
We analyze the security and privacy ability CFL achieves under the widespread internal semi-honest participants and external malicious adversaries threat models \cite{Bonawitz2019,Bonawitz2017,PPT2021}. In detail, the proposed CFL protocol can protect the privacy of the client's individual local contributions by a random noise, which is generated initially and eliminated ultimately. Besides, CFL guarantees the data integrity and authenticity of local contributions by leveraging the authenticated encryption scheme. Particularly, rigorous analyses show that CFL can guarantee the confidentiality of the communication key against hijacking attacks \cite{hu2019session}. In addition, we analyze the network connectivity and formally give the threshold of the key ring size, which proves that the proposed CFL protocol can guarantee secure transmission.

\noindent{\bfseries Experimental evaluation.}
We conduct experiments on the Trec06p and Trec07 datasets, which demonstrate that the proposed CFL protocol can ensure classification accuracy, model generalization, and rapid convergence rate of the global model. Also, the dropout-robustness of the system is achieved by CFL. More importantly, CFL improves communication efficiency by 43.25\% and computational efficiency by 0.87\% compared to PPT \cite{PPT2021}.

The rest of the paper is organized as follows. Section 2 reviews related works. In Section 3, we formalize the system model, security requirements, and design goal. Section 4 formally describes some preliminaries. Section 5 devises the proposed CFL protocol followed by security analysis, network connectivity analysis, and experiments in Sections 6 and7. We conclude this paper in Section 8.

\section{Related Works}
Existing decentralized FL works mainly focus on algorithms and frameworks. As for the algorithm, Baidu \cite{ring2017} proposed a bandwidth optimization algorithm, Ring AllReduce, which averages the gradient effectively on distributed devices. Yang \textit{et al.} \cite{Air2019} exploited the signal superposition property of wireless multiple-access channels to locally computing updates in FL by the over-the-air computation. Combining the PushSum gossip algorithm with stochastic gradient updates, Assran \textit{et al.} \cite{SGP2019} proposed SGP and OSGP to accelerate distributed training of deep neural networks. As for the framework, Ramanan \textit{et al.} \cite{Ramanan2020} proposed a blockchain-based aggregator-free FL framework, BAFFLE, which achieves high scalability and computational efficiency in a private Ethereum network. Hu \textit{et al.} \cite{Hu2019} addressed the decentralized FL problem by aggregating the local model segments based on the gossip protocol. Dubey \textit{et al.} \cite{Dubey2020} devised FEDUCB for FL in P2P networks to solve the contextual linear bandit problem. Li \textit{et al.} \cite{Li-Wen2020} proposed SimFL to build decision trees with bounded errors, which leverages locality-sensitive hashing to collect similarity information. However, the above works propose advanced decentralized FL frameworks or algorithms in a coarse-grained
manner. A formal description of decentralized FL settings and a general and fine-grained model training protocol that can guide model training for {\it FL in P2P networks} in practice are still absent.

Recently, Chen \textit{et al.} \cite{PPT2021} formalized the FL settings in P2P networks and proposed a fine-grained model training protocol, PPT, which aggregates clients' local contributions in a privacy-preserving manner. However, the PPT protocol is inefficient facing a large number of clients, i.e., {\it FL in large-scale P2P networks}. Besides, the communication key establishment phase in PPT is vulnerable to hijacking attacks \cite{hu2019session}, which might give rise to a security risk to the system.

To meet the efficiency requirement of FL in large-scale P2P networks, the proposed CFL hierarchically aggregates local contributions through the {\it multi cyclic transmission}. Also, the proposed CFL resists hijacking attacks by a proposed voting-based key revocation mechanism when establishing communication keys. To the best of our knowledge, the proposed CFL protocol is the first fine-grained model training protocol for FL in large-scale P2P networks, which is efficient and privacy-preserving.

\section{System Model, Security Requirements and Design Goal}

In this section, we formalize the system model, propose the security requirements, and identify the design goal.

\begin{figure*}[htbp]
	\centering
	\includegraphics[scale=0.24]{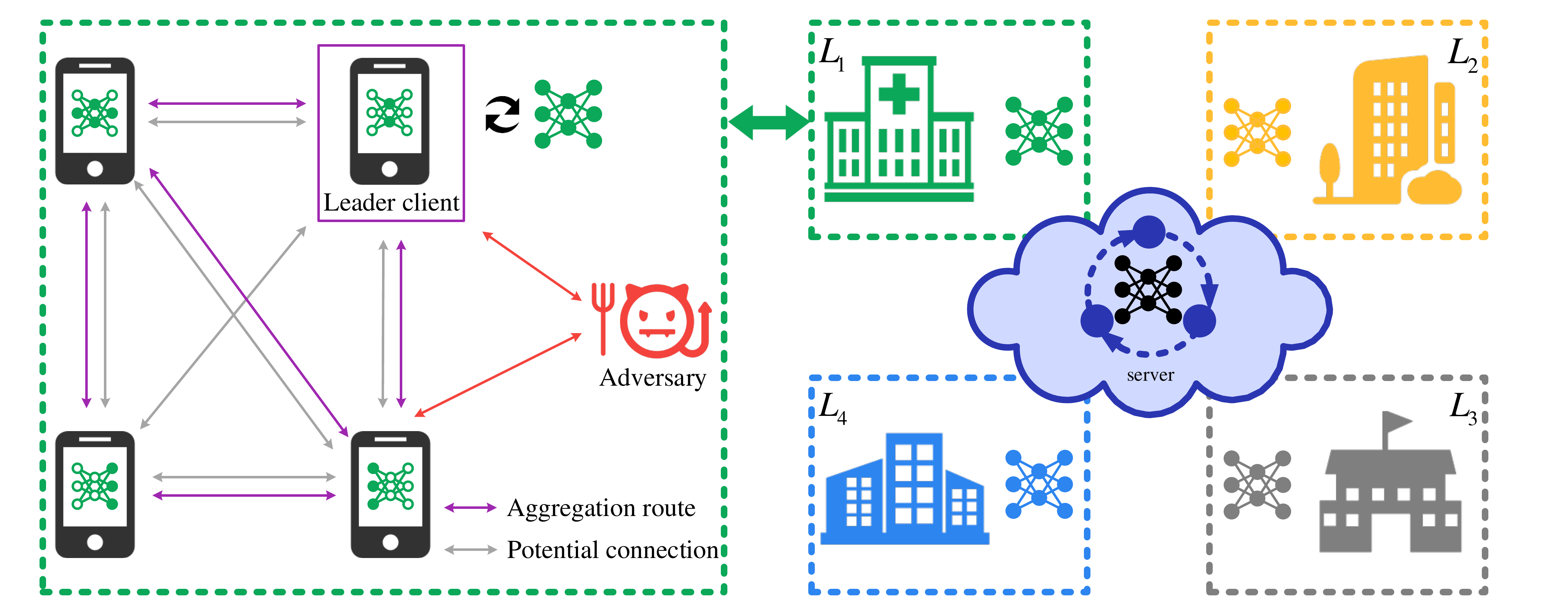}
	\caption{Federated learning in large-scale P2P networks} 
	\vspace{-1.5mm}
	\label{fig-1}  
\end{figure*}

\subsection{System Model}
In this paper, we consider the {\it FL in large-scale P2P networks}, which is an expansion of the system model of PPT \cite{PPT2021} in terms of client amount. Specifically, as shown in Fig. \ref{fig-1}, a large group of potential clients $ \mathbb{U}\!=\!\{u_{i}| i\!=\!1,2,\cdots,Z\} $ with constant wireless communication ranges are distributed in a large-scale P2P network. The clients truly participate in the global model training process are target clients, denoted by $ \mathbb{C}\!=\!\{c_{i}| i\!=\!1,2,\cdots,N\} $, $\mathbb{C}\! \subseteq \!\mathbb{U}$. Each target client $c_{i}$ has a local dataset $ \mathcal{D}_{i} $ of size $ \vert\mathcal{D}_{i}\vert $ containing private training data. Besides, a server loosely connected with a few potential clients is responsible for coordinating FL tasks, which means the server cannot aggregate clients' local contributions directly. The goal of target clients is to collaboratively train a global ML model $ W $ under the orchestration of the server, where $ W $ is usually expressed as a $d$-dimensional vector in the ML community.  

In detail, the collaborative training process is an iteration. For clarity, we take the $ t $-th round as an example to describe the iterative training process. The server first sends the global model $ W^{t} $ to the potential clients connected to it directly, and these potential clients pass $ W^{t} $ to other potential clients subsequently. After receiving the global model, each target client $ c_{i} $ initializes the local model as the global model $ W^{t} $, and trains it on $ \mathcal{D}_{i} $ to obtain an updated local model $ w_{i}^{t} $ through the stochastic gradient descent (SGD) algorithm, which is simply expressed as:
\begin{equation}\label{eq-1}
	w_{i}^{t}\leftarrow \textbf{\textit{Train}}(W^{t}, \mathcal D_{i}),
\end{equation}
where the $ w_{i}^{t} $ is a $d$-dimensional vector. Then, each target client $ c_{i} $ calculates its local model update parameters $ x_{i}^{t} $, shown as:
\begin{equation}\label{eq-2}
	x^{t}_{i}=w_{i}^{t}-W^{t}.
\end{equation}
Next, all target clients aggregate local model update parameters and upload the aggregation result to the server to update the global model, shown as:
\begin{equation}\label{eq-3}
	W^{t+1}=W^{t}+\sum_{i=1}^{N}p_{i}x^{t}_{i},
\end{equation}
where $ p_{i}=\frac{\vert\mathcal{D}_{i}\vert}{\sum\nolimits_{i=1}^{N}\vert\mathcal{D}_{i}\vert} $ is the aggregation weight of $c_{i}$, and $W^{t+1}$ is the updated global model. All participants under our consideration perform the above training process periodically until the global model converges to an optimal result, i.e., the final converged global model $ W^{*} $.

{\it Communication model.} Inherited from the communication characteristics of P2P networks, all data in the system is transmitted through public communication channels, which means the transmitted data could be eavesdropped by any client. In addition, all target clients can only communicate with their single-hop neighbor clients within their constant wireless communication ranges, and only a few clients could communicate with the server directly.

\subsection{Security Requirements}
In the context of {\it FL in large-scale P2P networks}, security and privacy are the most important concerns. Thus, we consider two widely used threat models in FL \cite{Bonawitz2019,Bonawitz2017,PPT2021}, i.e., the internal semi-honest participants threat model and the external malicious adversaries threat model. We further consider hijacking attacks \cite{hu2019session} from external malicious adversaries when establishing communication keys, which is thoughtless in PPT \cite{PPT2021}.

In the internal semi-honest (also named as honest-but-curious) participants threat model, the server and target clients perform prescribed operations honestly but are curious about others' local model update parameters, which means clients could infer private information from others' local model update parameters by executing model inversion attacks \cite{Fredrikson2015}. In particular, honest-but-curious potential clients may execute eavesdropping attacks \cite{attack2011} to obtain the private local model update parameters from target clients, thereby executing model inversion attacks \cite{Fredrikson2015} to threaten the privacy of local data.

As shown in Fig. \ref{fig-1}, in the external malicious adversaries threat model, a malicious adversary $ \mathcal{A} $ could execute tampering attacks \cite{attack2011} and impersonation attacks \cite{attack2011} to threaten the data integrity and the authenticity of aggregated local contributions. More seriously, $ \mathcal{A} $ could execute hijacking attacks, where $ \mathcal{A} $ hijacks honest participants to obtain the communication keys, and hence tampers with the aggregated model update parameters when aggregating local contributions, resulting in serious security risks. As a result, the proposed protocol should meet the following security requirements:
\begin{itemize}
\item{
	\textit{Privacy preservation.} Neither an honest-but-curious server nor clients can obtain others' precise local model update parameters. Also, the aggregated model update parameters should be only obtained by the sender and receiver but others during the transmission process.} 

\item{
	\textit{Data integrity and authenticity.} Ensure that local model update parameters are indeed sent by a legitimate client and have not been tampered with during transmission. That is, active attacks executed by $ \mathcal{A} $ should be resisted.}

\item{
	\textit{Confidentiality of the communication keys.} The communication key used to protect the aggregated model update parameters during transmission can only be obtained by the two corresponding target clients. In other words, hijacking attacks executed by $ \mathcal{A} $ should be detected and resisted.}
\end{itemize}

\subsection{Design Goal}
Under the system model and security requirements, our design goal is to propose an efficient and privacy-preserving global model training protocol for {\it FL in large-scale P2P networks}. Specifically, the following objectives should be achieved in the proposed protocol:

\begin{itemize}
	\item{{\it The security requirements should be guaranteed in the proposed protocol.} As mentioned in the security requirements, the proposed protocol should achieve privacy preservation, data integrity and authenticity, and the communication key confidentiality simultaneously under the internal semi-honest participants threat model and the external malicious adversaries threat model.}
	\item{{\it Communication and computational efficiency should be guaranteed in the proposed protocol.} Although more clients are involved in FL in large-scale P2P networks, the proposed protocol should ensure communication efficiency. In addition, the proposed protocol should reduce the computational cost of a single client to ensure the computational efficiency.}
\end{itemize}

\section{Preliminaries}

\subsection{Random Pairwise Keys Scheme}
The random pairwise keys scheme \cite{Chan2003} is initially deployed in distributed networks to establish communication keys. The basic idea is to associate two potential clients with a pairwise key based on their identities. The detailed steps of the random pairwise keys scheme are described as follows:

\begin{itemize}
	\item{{\it Step-1.} All potential clients generate their unique identities, denoted by $\textbf{\textit{ID}}\!=\!\{\mathit{ID}_{i}|i\!=\!1,2,\cdots,Z\}$.}
	\item{{\it Step-2.} Each potential client randomly matches its identity with $M$ other identities, and each client pair obtains a pair-wise $k$ from a key pool $\mathcal{K}$.}
	\item{{\it Step-3.} The pair-wise $k$ and the identity of the matching client are stored in both clients' key rings, denoted by $\textbf{\textit{R}}_{i}\!=\!\{(\mathit{ID}_{i\alpha}, k_{i\alpha})|\alpha\!=\!1,2,\cdots,M\}$.}
	\item{{\it Step-4.} Through the challenge-response mechanism, each potential client uses its key ring to establish communication keys with other potential clients that have common keys.}
\end{itemize}

\subsection{Authenticated Encryption}
Authenticated encryption \cite{ae2000} is an encryption scheme that provides both security and authentication, which comprises three algorithms: the key generation algorithm $\textbf{\textit{AE.Gen$\left(\cdot\right)$}}$, the encryption algorithm $\textbf{\textit{AE.Enc$\left(\cdot\right)$}}$, and the decryption algorithm $\textbf{\textit{AE.Dec$\left(\cdot\right)$}}$. The detailed steps of the authenticated encryption are as follows:

\noindent \textcircled{\oldstylenums{1}} $\textbf{\textit{AE.Gen$\left(\cdot\right)$}}$ generates a symmetric key $ K $, which is simply expressed as:
\begin{equation*}
	K\leftarrow\textbf{\textit{AE.Gen}}(1^{\kappa}),
\end{equation*}
where $ \kappa $ is a security parameter, i.e. the length of the symmetric key which is consistent with the length of each dimension of the local model.

\noindent \textcircled{\oldstylenums{2}} $\textbf{\textit{AE.Enc$\left(\cdot\right)$}}$ inputs the plaintext $ X $ and the key $ K $ and outputs the ciphertext $ Y $ and an authentication tag $ \sigma $ in the form of a message authentication code (MAC), which is simply expressed as:
\begin{equation*}
	(Y || \sigma)\leftarrow \textbf{\textit{AE.Enc}}(X, K) ,
\end{equation*}
where $||$ represents the string concatenation operator.

\noindent \textcircled{\oldstylenums{3}} $\textbf{\textit{AE.Dec$\left(\cdot\right)$}}$ inputs the ciphertext $ Y $ and the corresponding key $ K $ and outputs the plaintext $ X $, which is simply expressed as:
\begin{equation*}
	X\leftarrow \textbf{\textit{AE.Dec}}( (Y || \sigma), K) .
\end{equation*}  
If the $ \sigma $ fails authentication, $\textbf{\textit{AE.Dec$\left(\cdot\right)$}}$ outputs an unauthenticated symbol.

\section{System Design}
In this section, we propose an efficient and privacy-preserving \underline{C}luster \underline{F}ederated \underline{L}earning (CFL) global model training protocol for {\it FL in large-scale P2P networks}. The proposed CFL protocol mainly consists of the following five parts: cluster division, communication key establishment, local model training, aggregation within a single cluster, and aggregation and update across clusters. 

\subsection{Cluster Division}
There are large amounts of clients distributed in a large-scale P2P network under our consideration, where the complex connections among clients may bring huge communication latency if aggregating all local contributions directly. Therefore, we first divide all potential clients into several clusters $\mathbb{L}\!=\!\{L_{h}| h\!=\!1,2,\cdots\,\lambda\}$ based on their constant wireless communication ranges, shown in Fig. \ref{fig-1}. Note that, since only a few potential clients are directly connected to the server in each cluster, the potential clients and their connections within the cluster can also be regarded as a P2P network. Thus, the local contribution can be aggregated hierarchically, i.e., aggregation within a single cluster and aggregation across clusters subsequently. We stress that such a hierarchical aggregation process is named as {\it multi cyclic transmission}, since the aggregation processes within several clusters are performed simultaneously, which is more communication-efficient than {\it cyclic transmission} of the PPT protocol \cite{PPT2021}, where all participants aggregate local contributions in a cyclic manner directly.

For gravity, we regularize the symbols in a specific cluster $L_{h}\!\in\!\mathbb{L}$. The potential clients in $L_{h}$ are denoted by $ \textbf{\textit{U}}_{h}\!=\!\{u_{h,i}|i\!=\!1,2,\cdots,z_{h}\} $, and all $ \textbf{\textit{U}}_{h} \left(h=1,2,\cdots,\lambda\right) $ make up $ \mathbb{U} $, i.e., $\mathbb{U}=\bigcup\nolimits_{h=1}^{\lambda} \textbf{\textit{U}}_{h}$. The clients that actually participate in a specific round of global model training are target clients, denoted by $ \textbf{\textit{C}}_{h}\!=\!\{c_{h,i}|i\!=\!1,2,\cdots,n_{h}\} $, $\textbf{\textit{C}}_{h}\! \subseteq \!\textbf{\textit{U}}_{h}$, and all $\textbf{\textit{C}}_{h} \left(h=1,2,\cdots,\lambda\right)$ make up $ \mathbb{C} $, i.e., $\mathbb{C}=\bigcup\nolimits_{h=1}^{\lambda} \textbf{\textit{C}}_{h}$.

\begin{figure}[h]
	\centering
	\begin{tikzpicture}[scale=1]
		\path[fill=yellow!0, draw=black!50](0,0) rectangle (12,11.2);
		
		\node[right] at (0.5,10.9){\textbf{Secure Communication Key Establishment Protocol}};
		
		\draw[thick](0.6,10.7)-- (11.4,10.7);
		
		\node[right] at (0.5,10.4){\textbf{\small{$\bullet$  Communication key establishment:}}};
		
		\node[right] at (0.8,10){\small{- Each potential client $ u_{h,i} $ generates its identity $\mathit{ID}_{h,i}$. }};
		
		\node[right] at (0.8,9.6){\small{- Each identity is matched with $m_{h}$ other randomly selected identities.}};
		
		\node[right] at (0.8,9.2){\small{- Each client pair associated with two identities obtains a pairwise key }};
		
		\node[right] at (1.1,8.8){\small{from a key pool $\mathcal{K}$.}};
		
		\node[right] at (0.8,8.4){\small{- All identities and the corresponding keys make up $ u_{h,i} $'s key ring}};
		
		\node[right] at (1.1,8){\small{$\textit{\textbf{R}}_{h,i}\!=\!\{(\mathit{ID}_{h,i\alpha}, k_{h,i\alpha})|\alpha\!=\!1, 2, \cdots, m_{h}\}$.}};
		
		\node[right] at (0.8,7.6){\small{- $ u_{h,i} $  broadcasts a plaintext $a_{h,i}$ and $ m_{h} $ encrypted messages $ A_{h,i\alpha} $ }};
		
		\node[right] at (1.1,7.2){\small{using $ k_{h,i\alpha} $ in $\textbf{\textit{R}}_{h,i}$, shown as $ A_{h,i\alpha} \leftarrow  \textbf{\textit{CK.Enc}}(k_{h,i\alpha},a_{h,i}) $.}};
		
		\node[right] at (0.8,6.8){\small{- $ u_{h,j} $ ($j\!\neq\!i$) tries to decrypt $A_{h,i\alpha}$ using $ k_{h,j\alpha} $ in $\textbf{\textit{R}}_{h,j}$, which is shown}};
		
		\node[right] at (1.1,6.4){\small{as $ a_{h,i}^{\prime} \leftarrow \textbf{\textit{CK.Dec}}(k_{h,j\alpha},A_{h,i\alpha}) $.}};
		
		\node[right] at (0.8,6){\small{- $u_{h,j}$ compares $a_{h,i}^{\prime}$ with the plaintext $a_{h,i}$, seeking for shared keys.}};
		
		\node[right] at (0.8,5.6){\small{- $u_{h,i}$ and $u_{h,j}$ establish the communication key $\textit{K}_{hi,hj}$ by XOR.}};
		
		\node[right] at (0.5,5.2){\textbf{\small{$\bullet$  Key revocation:}}};
		
		\node[right] at (0.8,4.8){\small{- $u_{h,i}$ is assigned $m_{h}$ voting keys $\{v_{h,i\alpha}|\alpha\!=\!1,2,\!\cdots\!,\!m_{h}\}$ and the hash}};
		
		\node[right] at (1.1,4.4){\small{values of the voting keys of $m_{h}\!\!-\!1$ other voting members $\mathcal{H}(v_{h,i\beta})$,}};
		
		\node[right] at (1.1,4){\small{where $\mathcal{H}(\cdot)$ is the hash function, $\beta\!\neq\!\alpha$, and $1\!\leqslant\! \beta\!\leqslant\! m_{h}$.}};
		
		\node[right] at (0.8,3.6){\small{- When $u_{h,\mathcal{A}^{\prime}\alpha}$ votes against $ \mathcal{A}^{\prime} $, it broadcasts $v_{h,\mathcal{A}^{\prime}\alpha}$ in forms of plaintext.}};
		
		\node[right] at (0.8,3.2){\small{- $m_{h}\!-\!1$ other voting members verify the vote $v_{h,\mathcal{A}^{\prime}\alpha}$, shown as }};
		
		\node[right] at (1.1,2.8){\small{$\{0,1\} \leftarrow \operatorname{\textbf{\textit{Verify}}}\left(\mathcal{H}\left(v_{h,\mathcal{A}^{\prime}\alpha}\right), \mathcal{H}^{\prime}\left(v_{h,\mathcal{A}^{\prime}\alpha}\right)\right)$.}};
		
		\node[right] at (0.8,2.4){\small{- Once validated, the vote is marked, shown as $\mathcal{H}\left(v_{h,\mathcal{A}^{\prime}\alpha}\right) \leftarrow (v_{h,\mathcal{A}^{\prime}\alpha},\nu)$,}};
		
		\node[right] at (1.1,2){\small{where $\nu$ denotes the vote is marked.}};
		
		\node[right] at (0.8,1.6){\small{- If the number of marked votes is not less than $l_{h}$, $ \mathcal{A}^{\prime} $ will be revoked.}};
		
		\node[right] at (0.5,1.2){\textbf{\small{$\bullet$  Re-keying:}}};
		
		\node[right] at (0.8,0.8){\small{- The affected clients restart the \textbf{Communication key establishment}}};
		
		\node[right] at (1.1,0.4){\small{part to build new communication keys.}};

	\end{tikzpicture}
    \vspace{-1.5mm}
	\caption {Secure Communication Key Establishment Protocol.}\label{fig-2}
\end{figure}

\subsection{Communication Key Establishment}

To enhance the security while transmitting local model update parameters, we propose a \textbf{Secure Communication Key Establishment Protocol} in Fig. \ref{fig-2} to establish communication keys for potential clients within a single cluster, which is an important part of the proposed CFL protocol and comprises three parts, i.e., \emph{Communication key establishment}, \emph{Key revocation}, and \emph{Re-keying}. 

$\bullet$ \emph{Communication key establishment.}
Without loss of generality, we take the cluster $L_{h}$ as an example. Each potential client $ u_{h,i} $  first generates a unique identity $\mathit{ID}_{h,i}$. Then, $ u_{h,i} $ randomly matches its identity with $ m_{h} $ other identities, and each client pair obtains a pairwise key from a key pool $\mathcal{K}$ containing sufficient keys. Thus, $m_{h}$ sets of identities and the corresponding pairwise keys make up the key ring $\textit{\textbf{R}}_{h,i}\!=\!\{(\mathit{ID}_{h,i\alpha}, k_{h,i\alpha})|\alpha\!=\!1, 2, \cdots\!, m_{h}\}$ which is stored locally on $ u_{h,i} $. Note that the above offline operations do not consume communication resources.

Afterwards, all potential clients discover other potential clients that have shared keys with them through the challenge-response mechanism. Specifically, $u_{h,i}$ broadcasts a plaintext $a_{h,i}$ and $m_{h}$ messages $\{A_{h,i\alpha}|\alpha\!=\!1,2,\cdots,m_{h}\}$ encrypted by $m_{h}$ pairwise keys $\{k_{h,i\alpha}|\alpha\!=\!1,2,\cdots,m_{h}\}$ in $\textbf{\textit{R}}_{h,i}$, shown as:
\begin{equation}\label{eq-4}
	A_{h,i\alpha} \leftarrow \textbf{\textit{CK.Enc}} (k_{h,i\alpha},a_{h,i}),
\end{equation}
where $ \textbf{\textit{CK.Enc}}(\cdot) $ is the symmetric encryption algorithm. After that, every other potential clients $u_{h,j}$ $(j\!\neq\!i)$ attempts to decrypt the ciphertexts successively using the locally stored $m_{h}$ pairwise keys in its key ring, shown as:
\begin{equation}\label{eq-5}
	a_{h,i}^{\prime} \leftarrow \textbf{\textit{CK.Dec}}(k_{h,j\alpha},A_{h,i\alpha}),
\end{equation}
where $a_{h,i}^{\prime}$ is the decryption result.
Consequently, $u_{h,j}$ compares $a_{h,i}^{\prime}$ with the plaintext $a_{h,i}$, seeking for shared keys. Thus, $u_{h,i}$ and $u_{h,j}$ execute XOR operations for all shared keys to obtain a communication key $K_{hi,hj}$.

$\bullet$ \emph{Key revocation.}
In the external malicious threat model under our consideration, an adversary $ \mathcal{A} $ can hijack honest potential clients to obtain the communication keys, thereby threatening the security of the system. Thus, we design a compact and efficient \emph{key revocation} mechanism based on public voting. 

In detail, once a potential client observes active attacks from the client $ \mathcal{A}^{\prime} $ hijacked by $ \mathcal{A} $, it broadcasts a negative vote. We define the voting members as $\{u_{h,\mathcal{A}^{\prime}\alpha}|\alpha\!=\!1,2,\cdots,m_{h}\}$ who are single-hop neighbors of $ \mathcal{A}^{\prime} $. Once $ \mathcal{A}^{\prime} $ is voted against more than $ l_{h} $, all voting members would cut off all connections to $ \mathcal{A}^{\prime} $ for revocation, thereby eliminating the $\mathcal{A}$'s influence on the system. 

In particular, the \textit{neighborhood broadcast} \cite{broadcast2005} mechanism is used for every public vote in the network. The mechanism adopts broadcasts between neighbors to hierarchically make a consensus in the network. We stipulate that all voting members need to broadcast any public votes received to maximize the probability of successful transmission to adjacent voting members. All broadcasts are in the form of plaintext, as the voters' identities have no need to be protected. The detailed voting process is demonstrated in Fig. \ref{fig-2}.

$\bullet$ \emph{Re-keying.} To ensure the long-term security of the system, the communication keys used for a period of time need to be updated. When the keys expire, the \emph{Re-keying} part will start. That is, all affected clients restart the \textit{Communication key establishment} part to rebuild new communication keys.

\subsection{Local Model Training}
After establishing communication keys, all potential clients collaboratively train the global model. Taking the $ t $-th round as an example, the server first distributes a global model $ W^{t} $ to the potential clients directly connected to it. These potential clients subsequently pass $ W^{t} $ to other potential clients within the same cluster. After receiving $ W^{t} $, each target client $ c_{h,i} $ initializes the local model as $ W^{t} $, and trains it on its local dataset $ \mathcal D_{h,i}$, shown as:
\begin{equation}\label{eq-6}
	w_{h,i}^{t}\leftarrow \textbf{\textit{Train}}(W^{t},\mathcal D_{h,i}).
\end{equation}
Then, $c_{h,i}$ calculates the local model update parameters, shown as: 
\begin{equation} \label{eq-7}
	x_{h,i}^{t}=w_{h,i}^{t}-W^{t}.
\end{equation}
Afterwards, $ c_{h,i} $ calculates its weighted local model update parameters $ X_{h,i}^{t}=p_{h,i}x_{h,i}^{t} $, where $ p_{h,i}=\frac{\vert\mathcal{D}_{h,i}\vert}{\sum\nolimits_{i=1}^{n_{h}}\vert\mathcal{D}_{h,i}\vert} $.

\subsection{Aggregation within a Single Cluster }

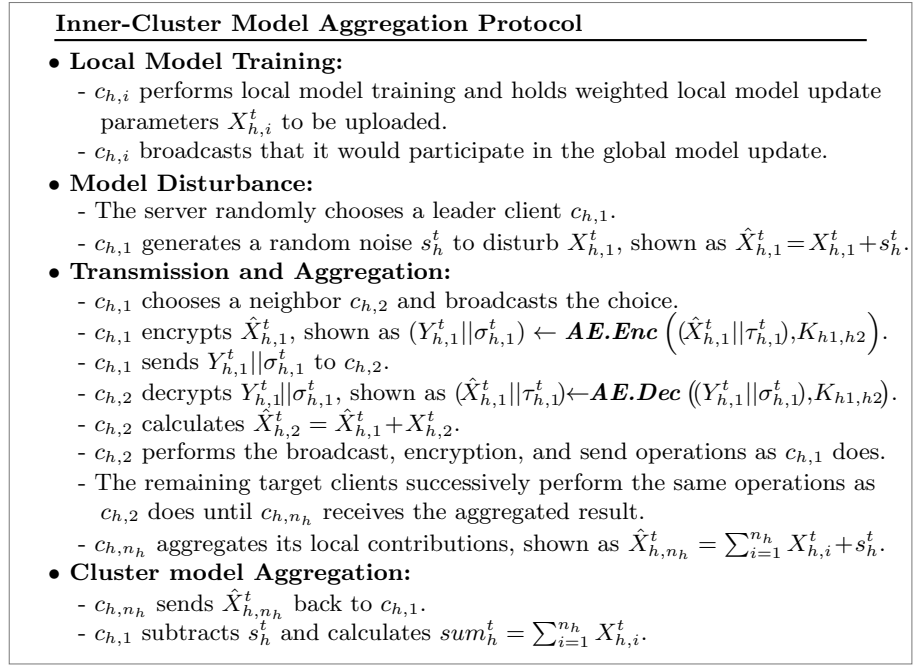
\begin{figure}[htbp]
	\centering
	\begin{tikzpicture}[scale=1.0]
		\path[fill=yellow!0, draw=black!50](0,0) rectangle (12,8.8);
		
		\node[right] at (0.5,8.5){\textbf{Inner-Cluster Model Aggregation Protocol }};
		
		\draw[thick](0.6,8.3)-- (11.4,8.3);
		
		\node[right] at (0.4,8){\small{\textbf{$\bullet$ Local Model Training:}}};
		
		\node[right] at (0.8,7.6){\small{- $c_{h,i}$ performs local model training and holds weighted local model update}};
		
		\node[right] at (1.1,7.2){\small{parameters $X_{h,i}^{t}$ to be uploaded.}};
		
		\node[right] at (0.8,6.8){\small{- $c_{h,i}$ broadcasts that it would participate in the global model update.}};		
		
		\node[right] at (0.4,6.4){\small{\textbf{$\bullet$ Model Disturbance:}}};
		
		\node[right] at (0.8,6){\small{- The server randomly chooses a leader client $c_{h,1}$.}};
		
		\node[right] at (0.8,5.6){\small{- $c_{h,1}$ generates a random noise $s_{h}^{t}$ to disturb $ X_{h,1}^{t} $, shown as $\hat{X}_{h,1}^{t}\!=\! X_{h,1}^{t}\!+\!s_{h}^{t}$.}};		
		
		\node[right] at (0.4,5.2){\small{\textbf{$\bullet$ Transmission and Aggregation:}}};
		
		\node[right] at (0.8,4.8){\small{- $c_{h,1}$ chooses a neighbor $c_{h,2}$ and broadcasts the choice.}};		
		
		\node[right] at (0.8,4.4){\small{- $c_{h,1}$ encrypts $\hat{X}_{h,1}^{t}$, shown as $ (Y_{h,1}^{t}||\sigma_{h,1}^{t}) \leftarrow \textbf{\textit{AE.Enc}}\left((\!\hat{X}_{h,1}^{t}||\tau_{h,1}^{t}\!),\! K_{h1,h2}\right) $.}};
		
		\node[right] at (0.8,4){\small{- $c_{h,1}$ sends $ Y_{h,1}^{t}||\sigma_{h,1}^{t} $ to $ c_{h,2} $.}};
		
		\node[right] at (0.8,3.6){\small{- $c_{h,2}$ decrypts $ Y_{h,1}^{t}\!||\sigma_{h,1}^{t} $, shown as $(\!\hat{X}_{h,1}^{t}||\tau_{h,1}^{t}\!)\!\! \leftarrow\!\!  \textbf{\textit{AE.Dec}}\left(\!(Y_{h,1}^{t}||\sigma_{h,1}^{t}\!),\!K_{h1,h2}\!\right) $.}};
		
		\node[right] at (0.8,3.2){\small{- $c_{h,2}$ calculates $\hat{X}_{h,2}^{t}= \hat{X}_{h,1}^{t}\!+\!X_{h,2}^{t}$.}};
		
		\node[right] at (0.8,2.8){\small{- $c_{h,2}$ performs the broadcast, encryption, and send operations as $c_{h,1}$ does. }};
		
		\node[right] at (0.8,2.4){\small{- The remaining target clients successively perform the same operations as }};
		
		\node[right] at (1.1,2){\small{$c_{h,2}$ does until $c_{h,n_{h}}$ receives the aggregated result.}};
		
		\node[right] at (0.8,1.6){\small{- $c_{h,n_{h}}$ aggregates its local contributions, shown as $\hat{X}_{h,n_{h}}^{t}= \sum\nolimits_{i=1}^{n_{h}}X_{h,i}^{t}\!+\!s_{h}^{t}$.  }};		
		
		\node[right] at (0.4,1.2){\small{\textbf{$\bullet$ Cluster model Aggregation:}}};	
		
		\node[right] at (0.8,0.8){\small{- $c_{h,n_{h}}$ sends $\hat{X}_{h,n_{h}}^{t}$ back to $c_{h,1}$.   }};
		
		\node[right] at (0.8,0.4){\small{- $c_{h,1}$ subtracts $s_{h}^{t}$ and calculates $sum_{h}^{t} = \sum\nolimits_{i=1}^{n_{h}}X_{h,i}^{t}$.}};
		
	\end{tikzpicture}
    \vspace{-1.5mm}
	\caption {Inner-Cluster Model Aggregation Protocol.}
	\label{fig-3}
\end{figure}

After local training, all target clients aggregate their local model update parameters within their clusters respectively. The server randomly selects a target client directly connected to it as the leader client, namely $ c_{h,1} $. Then, $ c_{h,1} $ generates a random noise $s_{h}^{t}$ to disturb its weighted local model update parameters $X_{h,1}^{t}$, shown as:
\begin{equation}\label{eq-8}
	\hat{X}_{h,1}^{t}= X_{h,1}^{t}+s_{h}^{t},
\end{equation}
where $s_{h}^{t}$ is a $ d $-dimensional vector. Subsequently, $ c_{h,1} $ attaches a timestamp $\tau_{h,1}^{t}$ to the disturbed model update parameters, which is represented as $ \hat{X}_{h,1}^{t}||\tau_{h,1}^{t} $. Next, $c_{h,1}$ randomly selects a neighbor client $ c_{h,2} $ and encrypts $ \hat{X}_{h,1}^{t}||\tau_{h,1}^{t} $ using their communication key $ K_{h1,h2} $, shown as:
\begin{equation}\label{eq-9}
	(Y_{h,1}^{t}||\sigma_{h,1}^{t}) \leftarrow \textbf{\textit{AE.Enc}}\left((\hat{X}_{h,1}^{t}||\tau_{h,1}^{t}), K_{h1,h2}\right).
\end{equation}
Then, $ c_{h,1} $ sends $ Y_{h,1}^{t}||\sigma_{h,1}^{t} $ to $ c_{h,2} $, and $ c_{h,2} $ recovers it subsequently, shown as:
\begin{equation}\label{eq-10}
	(\hat{X}_{h,1}^{t}||\tau_{h,1}^{t}) \leftarrow  \textbf{\textit{AE.Dec}}\left((Y_{h,1}^{t}||\sigma_{h,1}^{t}),K_{h1,h2}\right) .
\end{equation}
After validating the timestamp $ \tau_{h,1}^{t} $, $c_{h,2}$ aggregates its own weighted model update parameters, shown as:
\begin{equation}\label{eq-11}
	\hat{X}_{h,2}^{t}= \hat{X}_{h,1}^{t}+X_{h,2}^{t}. 
\end{equation}

Subsequently, $ c_{h,2} $ selects a new neighbor client $ c_{h,3} $, and performs the encryption and transmission operations as $ c_{h,1} $ does. The remaining target clients perform the same decryption, aggregation, encryption, and transmission operations successively until all the target clients are traversed. Note that the aggregation route follows the \textit{depth-first traversal algorithm in the graph} \cite{Depth72}, where each target client prefers to select the target clients that have not yet participated in the aggregation. Finally, the final client $ c_{h,n_{h}} $ sends the aggregated model update parameters back to $c_{h,1}$. $c_{h,1}$ subtracts the noise $s_{h}^{t}$ to obtain the aggregation result within its cluster, denoted by $ sum_{h}^{t} $ and defined as:
\begin{equation}\label{eq-12}
	sum_{h}^{t} = \sum\limits_{i=1}^{n_{h}}X_{h,i}^{t}. 
\end{equation}

To make a consensus on the aggregation route, we also adopt the \textit{neighborhood broadcast} \cite{broadcast2005} mechanism during transmission, where each target client hierarchically broadcasts its behaviors to all clients. With the deployment of the \textit{neighborhood broadcast} mechanism, dropout-robustness is guaranteed.

For gravity, we illustrate the model aggregation process within a single cluster as the \textbf{Inner-Cluster Model Aggregation Protocol} in Fig. \ref{fig-3}.

\subsection{ Aggregation and Update Across Clusters }
After finishing the model aggregation within each single cluster, each leader client $ c_{h,1} $ could obtain the inner-cluster aggregation result $ sum_{h}^{t} $. Next, each leader client uploads $ sum_{h}^{t} $ directly to the server, and the server aggregates all $ sum_{h}^{t},h=1,\cdots,\lambda $, shown as:
\begin{equation}\label{eq-13}
	\mathit{SUM}^{t} = \sum\limits_{h=1}^{\lambda} q_{h}sum_h^{t}, 
\end{equation}
where $ q_{h}=\frac{\sum\nolimits_{D_{i}\in L_{h}}\vert\mathcal{D}_{i}\vert}{\sum\nolimits_{i=1}^{N}\vert\mathcal{D}_{i}\vert} $ is the aggregation weight of the cluster $L_{h}$. Consequently, the server updates the global model, shown as:
\begin{equation}\label{eq-14}
	W^{t+1}=W^{t}+\mathit{SUM^{t}}.
\end{equation}

Finally, the server distributes the updated global model to all potential clients, and the target clients iteratively train and aggregate until the global model converges to $W^{*}$. 

For clarity of description, we illustrate the integrated CFL protocol from a high-level view in Fig. \ref{fig-4}.
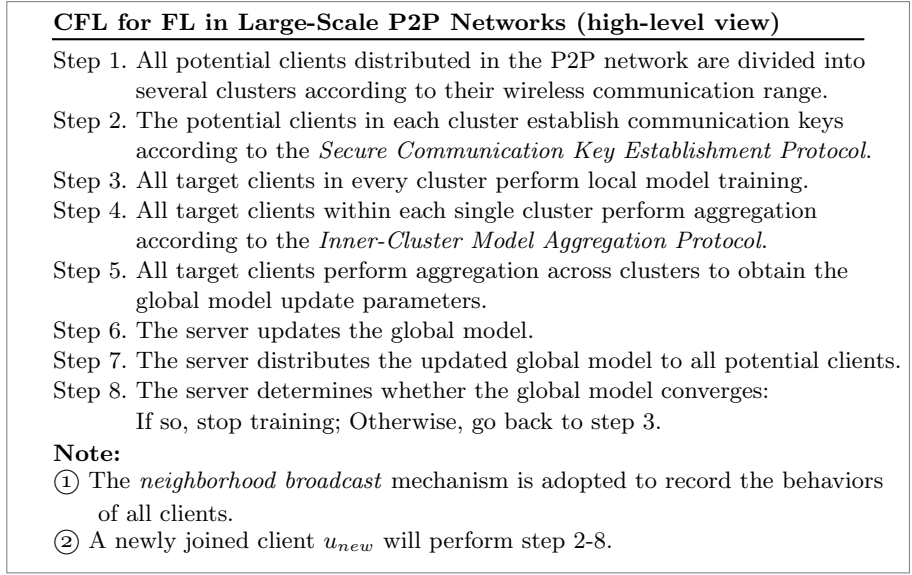
\begin{figure}
	\centering
	\begin{tikzpicture}[scale=1]
		\path[fill=yellow!0, draw=black!50](0,0) rectangle (12,7.6);
		
		\node[right] at (0.5,7.3){\textbf{CFL for FL in Large-Scale P2P Networks (high-level view)}};
		
		\draw[thick](0.6,7.1)-- (11.4,7.1);
		
		\node[right] at (0.5,6.8){\small{Step 1. All potential clients distributed in the P2P network are divided into}};
		
		\node[right] at (1.6,6.4){\small{several clusters according to their wireless communication range.  }};		
		
		\node[right] at (0.5,6){\small{Step 2. The potential clients in each cluster establish communication keys}};
		
		\node[right] at (1.6,5.6){\small{according to the \textit{Secure Communication Key Establishment Protocol}. }};
		
		\node[right] at (0.5,5.2){\small{Step 3. All target clients in every cluster perform local model training.}};
		
		\node[right] at (0.5,4.8){\small{Step 4. All target clients within each single cluster perform aggregation}};
		
		\node[right] at (1.6,4.4){\small{according to the \textit{Inner-Cluster Model Aggregation Protocol}.}};		
		
		\node[right] at (0.5,4){\small{Step 5. All target clients perform aggregation across clusters to obtain the}};
		
		\node[right] at (1.6,3.6){\small{global model update parameters.}};
		
		\node[right] at (0.5,3.2){\small{Step 6. The server updates the global model.}};
		
		\node[right] at (0.5,2.8){\small{Step 7. The server distributes the updated global model to all potential clients.}};
		
		\node[right] at (0.5,2.4){\small{Step 8. The server determines whether the global model converges:}};
		
		\node[right] at (1.6,2){\small{If so, stop training;  Otherwise, go back to step 3.}};
		
		\node[right] at (0.5,1.6){\small{\textbf{Note:}}};
		
		\node[right] at (0.5,1.2){\small{\noindent \textcircled{\oldstylenums{1}} The \emph{neighborhood broadcast} mechanism is adopted to record the behaviors}};
		
		\node[right] at (1.1,0.8){\small{of all clients. }};
		
		\node[right] at (0.5,0.4){\small{\noindent \textcircled{\oldstylenums{2}}  A newly joined client $u_{new}$ will perform step 2-8.}};
		
	\end{tikzpicture}
    \vspace{-1.5mm}
	\caption {High-level view of the CFL protocol.}
	\label{fig-4}
\end{figure}

\section{Analysis}
In this section, we analyze the security properties of the proposed CFL protocol under the internal semi-honest participants threat model and the external malicious adversaries threat model. Particularly, we analyze the network connectivity, which is critical for secure communication among clients. 

\subsection{Security Analysis}
In the internal semi-honest participants threat model, we mainly focus on the privacy requirements facing the honest-but-curious server, target clients, and potential clients.

\textit{The proposed CFL protocol is privacy-preserving facing an honest-but-curious server.} An honest-but-curious server may infer the private local training data from clients' local contributions by executing model inversion attacks \cite{Fredrikson2015}. To defend against such inferences, the proposed CFL protocol aggregates local contributions in a privacy-preserving manner, where the server can only obtain the aggregation result rather than the private local model update parameters.   

\textit{The proposed CFL protocol is privacy-preserving facing honest-but-curious target clients.} An honest-but-curious target client may also execute model inversion attacks \cite{Fredrikson2015} to threaten the privacy contained in the local datasets. To resist such attacks, a random noise is introduced to disturb the aggregated model update parameters. Therefore, all target clients can only get disturbed local contributions or the intermediate aggregation results, rather than the precise local model update parameters. 

\textit{The proposed CFL protocol is privacy-preserving facing honest-but-curious potential clients.} An honest-but-curious potential client may execute eavesdropping attacks \cite{attack2011} to obtain the intermediate aggregation result. Thus, the proposed CFL protocol adopts the authenticated encryption to protect aggregated model update parameters. Therefore, an honest-but-curious potential client cannot eavesdrop the plaintext of other's local model update parameters.

Furthermore, in the external malicious threat models, we discuss the data integrity and authenticity of the aggregated model update parameters and the confidentiality of the communication keys.

\textit{The proposed CFL protocol guarantees the data integrity and authenticity of the aggregated model update parameters.} A malicious adversary can execute tampering attacks \cite{attack2011} and impersonation attacks \cite{attack2011} to threaten the data integrity and authenticity of the aggregated model update parameters. To defend against such attacks, the proposed CFL protocol adopts authenticated encryption to protect aggregated model update parameters, where the authentication tag can effectively resist unauthorized tampering of messages and verify whether the sender is a trusted source. 

\textit{The proposed CFL protocol guarantees the confidentiality of the communication keys.} More seriously, a malicious adversary can hijack honest participants, which may lead to the leakage of the communication keys, resulting in various attacks to threaten the security of the system. Thus, the {\bf Secure Communication Key Establishment Protocol} is enhanced by a voting-based \emph{key revocation} mechanism to remove the contaminated keys, thereby eliminating the possibility of malicious adversaries stealing communication keys. Therefore, the confidentiality of the communication key is ensured. 

\subsection{Network Connectivity Analysis}
In the proposed CFL protocol, we use the \textbf{Secure Communication Key Establishment Protocol} to establish communication keys for potential clients. We analyze the network connectivity based on the established communication keys in \textbf{Theorem \ref{thm-1}}, which is the basis of secure transmission.

\begin{theorem}\label{thm-1}
	For a cluster $ L_{h} $, the size of the key ring for each potential client only needs to be $ m_{h} $ to ensure a high network connectivity probability.
\end{theorem}

\begin{proof}
	We first formalize the network topology in cluster $ L_{h} $ as a random graph $ G (n_{h}, r_{h}) $, where $n_{h}$ is the amount of potential clients in $ L_{h} $, and $r_{h}\in[0,1]$ is the probability that any two potential clients are connected. 
	
	To reflect the desired connectivity probability $P_{c}$ of $ G (n_{h}, r_{h}) $, we recall \textbf{Lemma \ref{lem-1}}.
	
	\begin{lemma}[\cite{Boyd2006}] \label{lem-1}
		Given a desired connectivity probability $P_{c}$ for a graph $ G(n_{h},r_{h}) $, the threshold function $r_{h}$ is defined by:
		\begin{equation}\label{eq-15}
			P_{c}=\lim _{n_{h} \rightarrow \infty} \operatorname{Pr}[G(n_{h}, r_{h}) \text { is connected }]=e^{-e^{-c}},
		\end{equation}
		where $ r_{h}=\frac{\ln(n_{h})}{n_{h}}+\frac{c}{n_{h}} $, and $c$ is any real constant.
	\end{lemma}
	
	According to \textbf{Lemma \ref{lem-1}}, given a high connectivity probability (e.g., 0.999), we can calculate $r_{h}$ as follows:
	\begin{equation}\label{eq-16}
		r_{h}= \frac{\ln (n_{h})}{n_{h}}+\frac{-\ln(-\ln(P_{c}))}{n_{h}}.
	\end{equation}
	
	Thus, we can infer that a graph with $n_{h}$ vertices has at most $\binom{n_{h}}{2}$ edges \cite{Bondy1976}. Assume that $G (n_{h}, r_{h})$ has $e$ edges, we can calculate $ r_{h} $ as follows:
	\begin{equation}\label{eq-17}
		r_{h}=\frac{e}{\binom{n_{h}}{2}}.
	\end{equation}
	
	Next, we can calculate the key ring size $m_{h}$ of the potential client, i.e., the expected degree of $G (n_{h}, r_{h})$, shown as:
	\begin{equation}\label{eq-18}
		m_{h}=\frac{2e}{n_{h}}.
	\end{equation}
	
	Therefore, according to the equations (\ref{eq-17}) and (\ref{eq-18}), we can figure out the relationship between $ m_{h} $ and $ r_{h} $ as:
	\begin{equation}\label{eq-19}
		m_{h}=r_{h}*(n_{h}-1)\sim r_{h}*n_{h},
	\end{equation} 
	where $\sim$ represents an equivalence symbol.
	
	As a result, according to the equation (\ref{eq-16}), we can calculate the key ring size $ m_{h} $ as: 
	\begin{equation}\label{eq-20}
		m_{h}=\ln (n_{h})-\ln(-\ln(P_{c})),
	\end{equation}
which complete the proof.
	$\hfill\blacksquare$
\end{proof}

\textbf{Theorem \ref{thm-1}} indicates that the size of the key ring of each potential client only needs to be $ m_{h}=\ln (n_{h})-\ln(-\ln(P_{c}))$ for a high network connectivity probability. In other words, given the size of the key ring $ m_{h} $, all aggregation results can be transmitted in ciphertext with a high probability.    

\section{Experiments and Evaluation}
In this section, we conduct the experiments in a spam classification scenario. All experiments are implemented on the same computing environment (Linux Ubuntu 16.04, Intel i7-6950X CPU, 64 GB RAM and 5TB SSD) with Tensorflow, Keras and PyCryptodome.

\subsection{Experimental Design}

\noindent{\bfseries Dataset.} 
The datasets used in our experiments consist of Trec06p \footnote{We give the download link as: https://trec.nist.gov/data/spam.html.\label{web}} and Trec07 \textsuperscript{\ref {web}}, which are two English e-mail datasets from the real world. The Trec06p dataset contains 12910 hams and 24912 spams in the main corpus with messages, and the Trec07 dataset contains 25220 hams and 50199 spams. 

\noindent{\bfseries Network connectivity topology.} 
As shown in Fig. \ref{fig-5}, in our experiments, there are 200 potential clients distributed in a large-scale P2P network which is divided into five clusters. Among the potential clients, 100 target clients are illustrated as the gray nodes, and the black nodes are leader clients. Further, we illustrate the connections among potential clients based on the communication keys as the gray lines and the aggregation route as the red lines.

\begin{figure}[htb] 
	\centering 
	\includegraphics[scale=0.1]{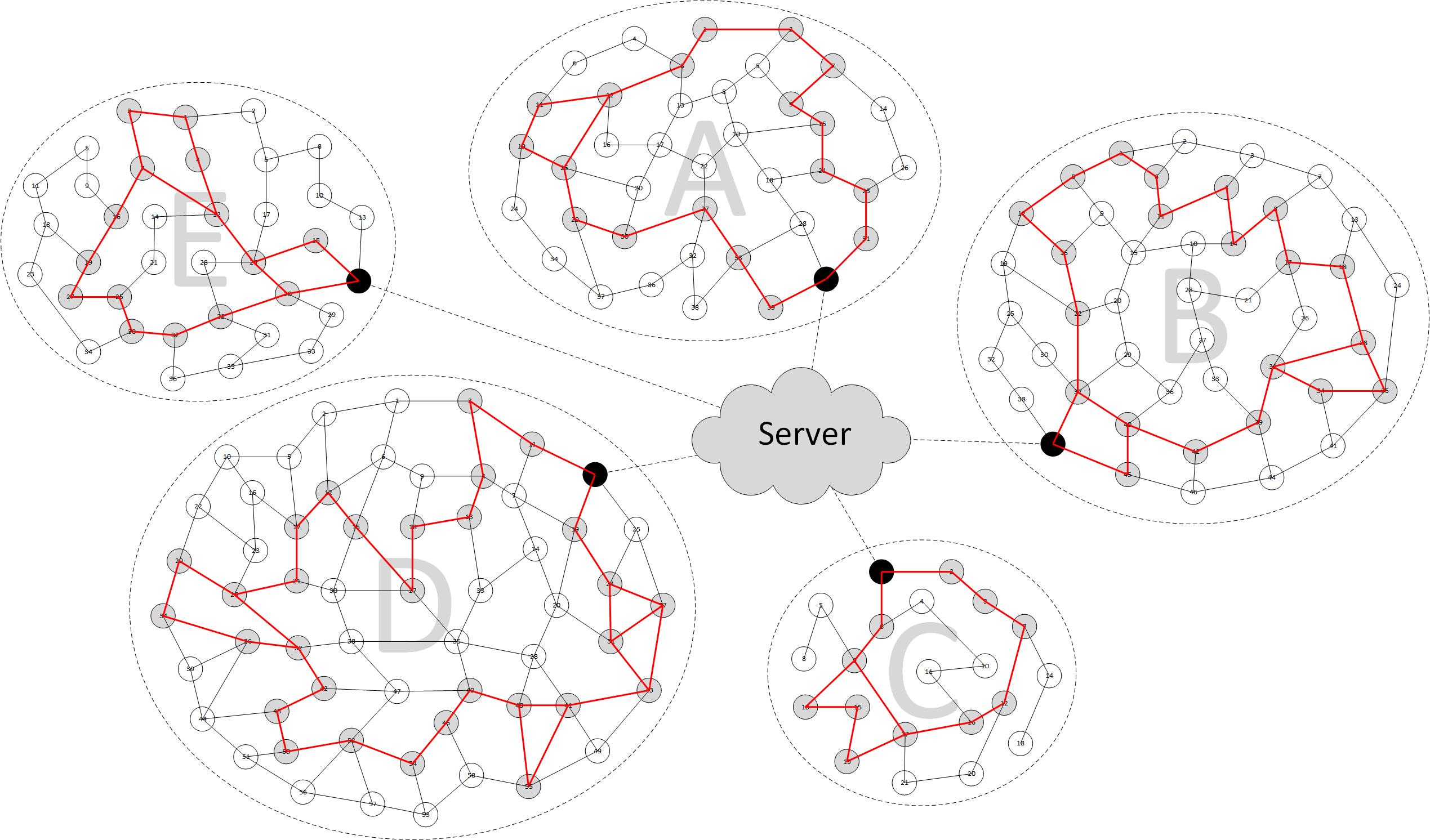} 
	\caption{Network connectivity topology.} 
	\label{fig-5}
\end{figure}

\noindent{\bfseries Parameter tuning.} 
In our experiments, all target clients use the local datasets to collaboratively train a global model according to the proposed CFL protocol. It is worth mentioning that the Trec06p dataset is divided into the Trec06p training set and the Trec06p testing set according to the ratio of 3:1, and the Trec06p training set is adopted to train the original global model. In addition, the Trec07 dataset is divided into two parts. One part is further divided into 100 local datasets for target clients, and the set size follows a normal distribution of mean 600 and variance 100. The other part serves as the Trec07 testing set to evaluate the global model performance. The original global model and local models have the same structure, which consists of two convolution layers, two pooling layers, and three fully connected layers. We also set the loss function as the \textit{cross-entropy error} and the active function as the \textit{ReLU}. The gradient descent algorithm is set as SGD with the learning rate of 0.1. Besides, we use the \textit{AES-GCM-128bit} algorithm \cite{galois2004} for authenticated encryption.  

\subsection{Experimental Results and Evaluation}
In order to validate the model performance of the proposed CFL protocol, we evaluate the accuracy and loss value of the global model in Fig. \ref{fig-6}. After 14 rounds of global model training, the global model converges, and the accuracy and loss value of the final global model are 99.32 \%  and 0.0356 on the Trec07 testing set. The results show that the proposed CFL protocol can guarantee the convergence of the global model with high classification accuracy.  

\begin{figure}[htbp]
	\centering
	\subfigure[Accuracy of global model.]{
		\begin{minipage}[t]{0.45\linewidth}
			\centering
			\includegraphics[scale=0.24,trim=10 0 0 10,clip]{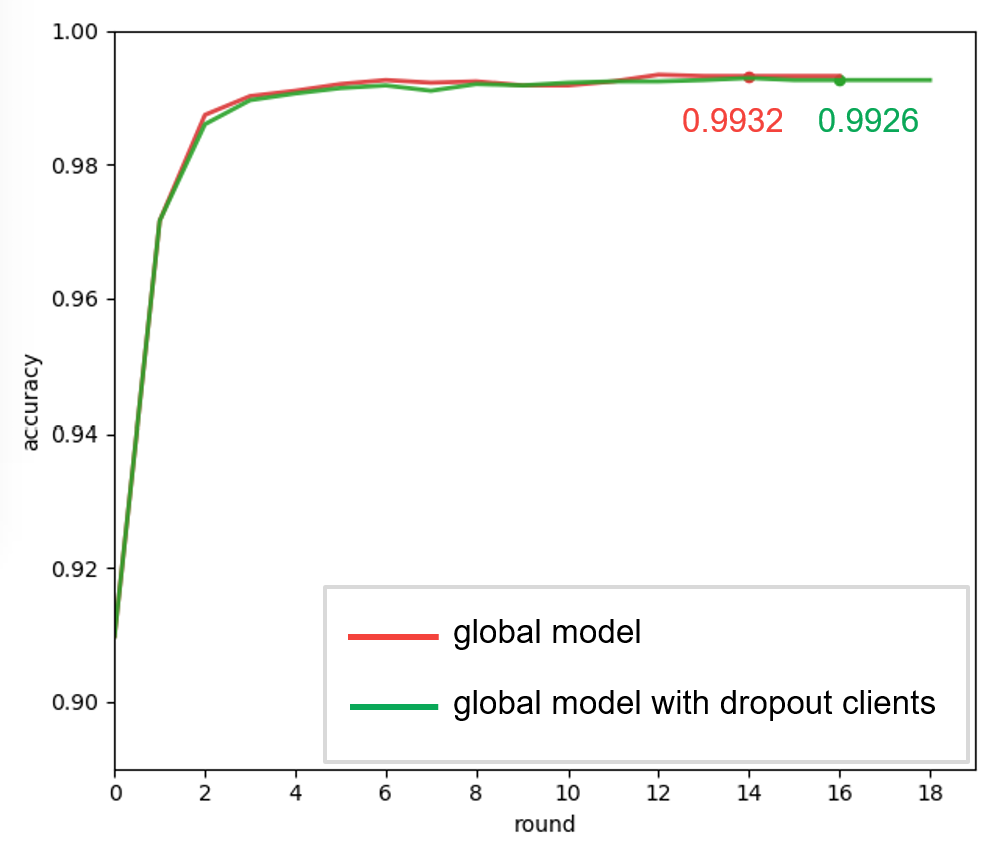}
		\end{minipage}%
	}%
	\subfigure[Loss value of global model.]{
		\begin{minipage}[t]{0.45\linewidth}
			\centering
			\includegraphics[scale=0.24,trim=10 0 0 10,clip]{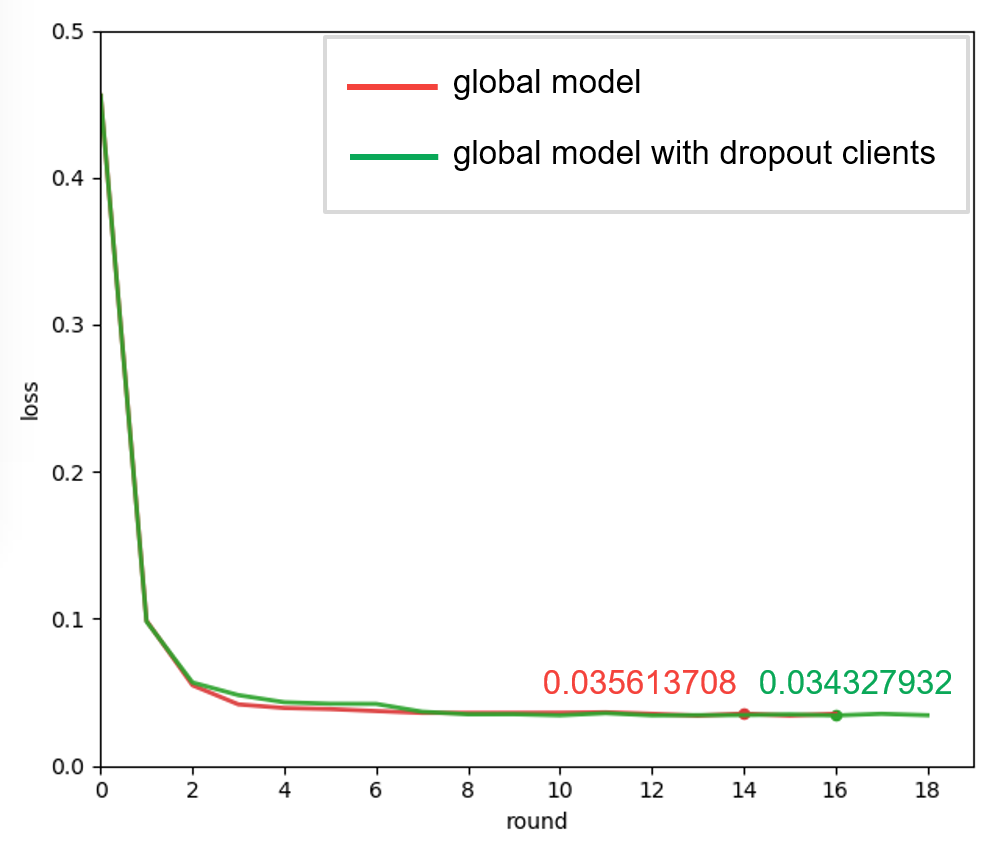}
		\end{minipage}%
	}%
	\centering
	\vspace{-1.5mm}
	\caption{Global model performance on the Trec07 testing set.}
	\label{fig-6}
\end{figure}

Moreover, we compare the performances of the original global model and the final global model by illustrating the receiver operating characteristic (ROC) curve and the area under the curve (AUC) on both Trec06p and Trec07 testing sets in Fig. \ref{fig-7}. The AUC of the final global model is larger than that of the original global model on the Trec07 testing sets and is close to the AUC on the Trec06p testing set, which shows that the final global model outperforms the original global model in terms of model generalization.

\begin{figure}[htbp]
	\centering
	\subfigure[Trec06p.]{
		\includegraphics[scale=0.115]{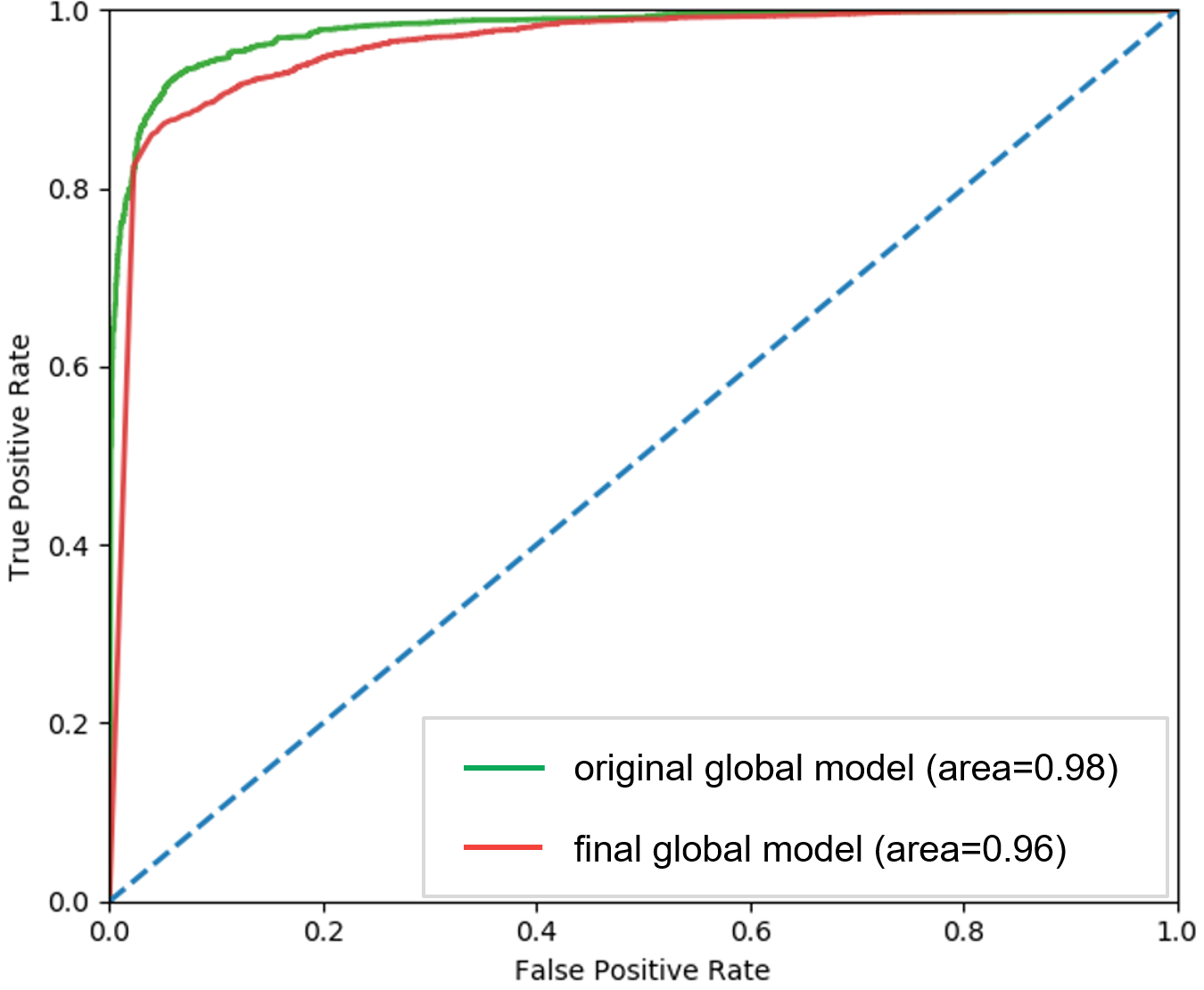}}
	\subfigure[Trec06p-drop.]{
		\includegraphics[scale=0.115]{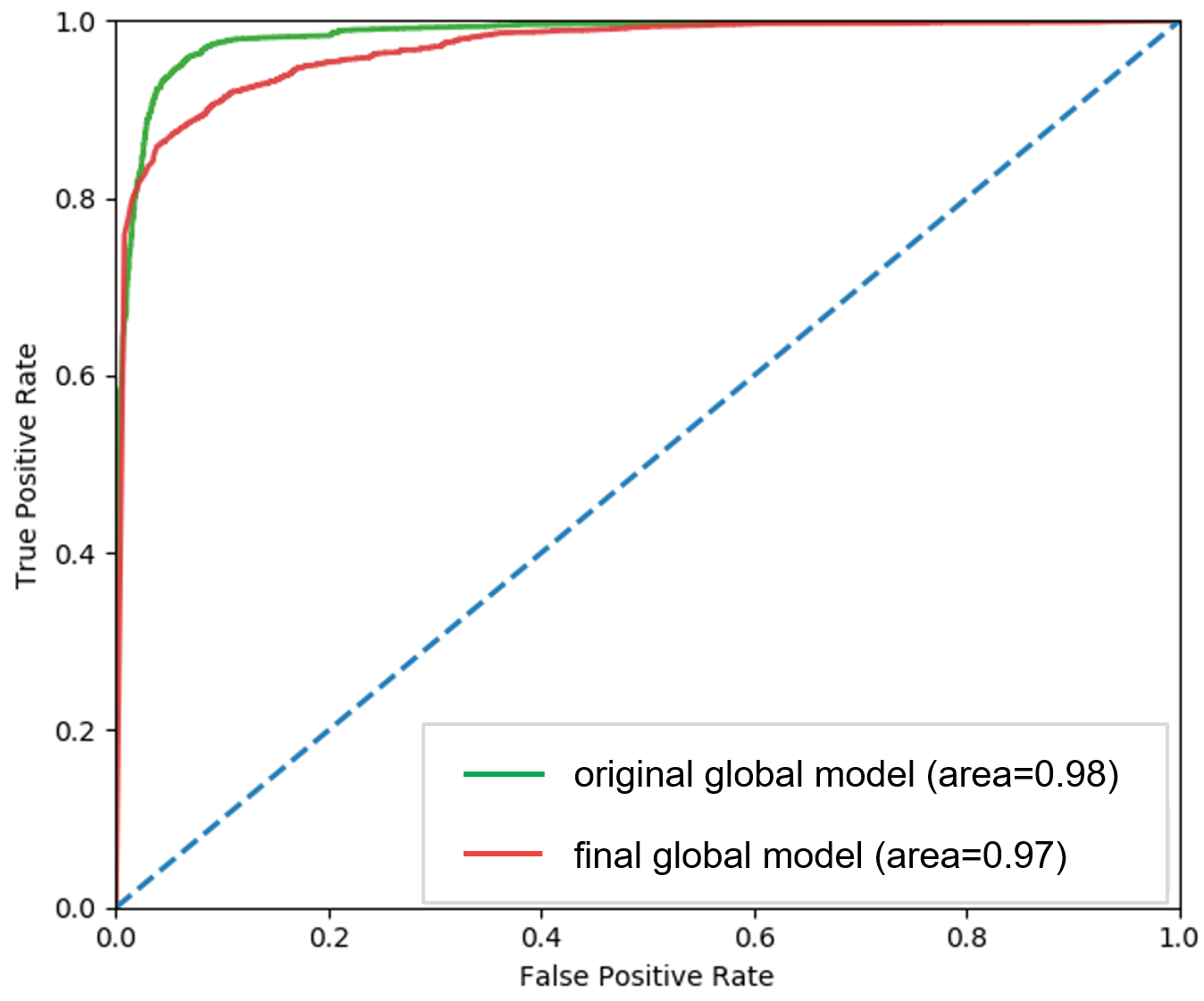}}
	\subfigure[Trec07.]{
		\includegraphics[scale=0.115]{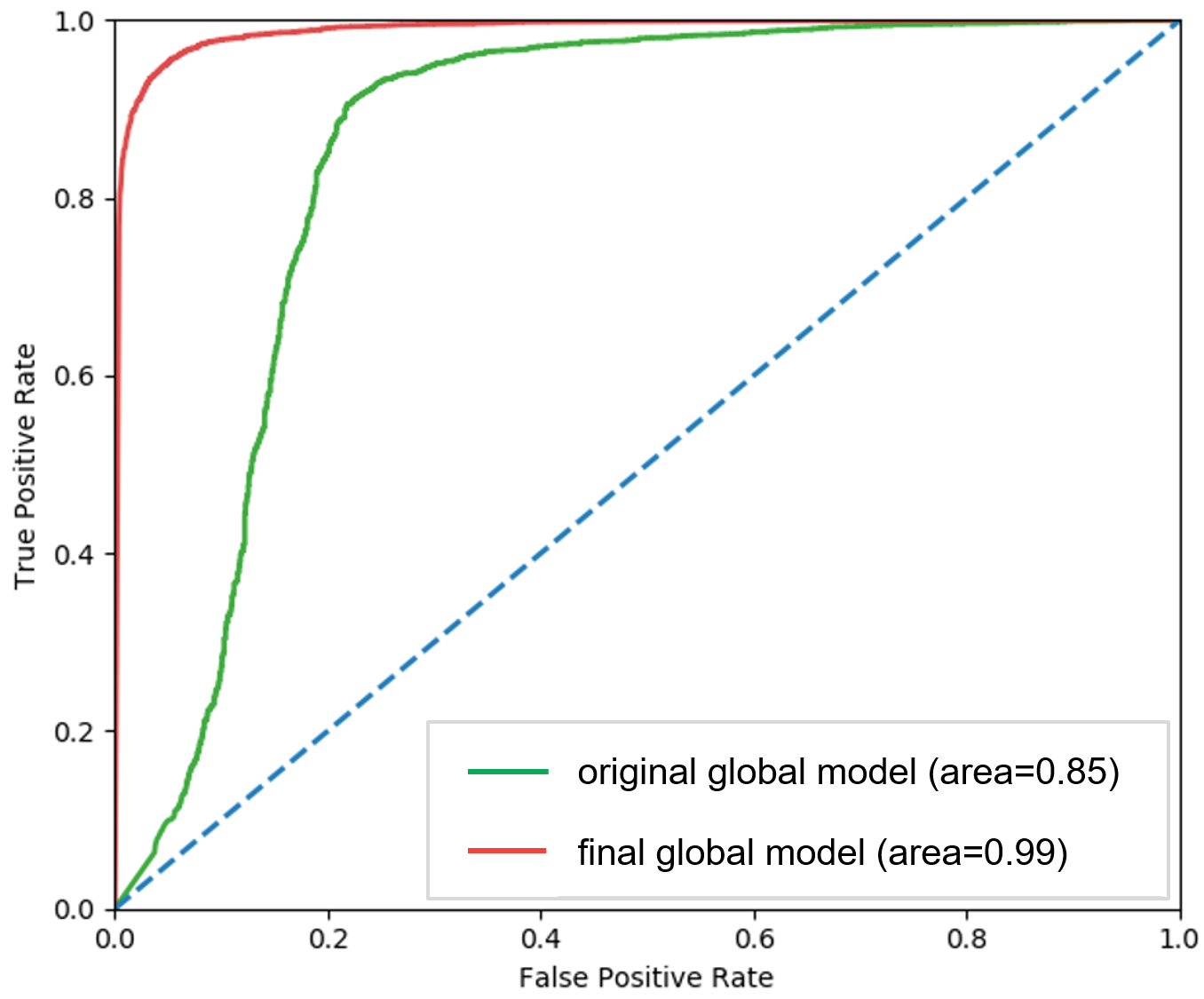}}
	\subfigure[Trec07-drop.]{
		\includegraphics[scale=0.115]{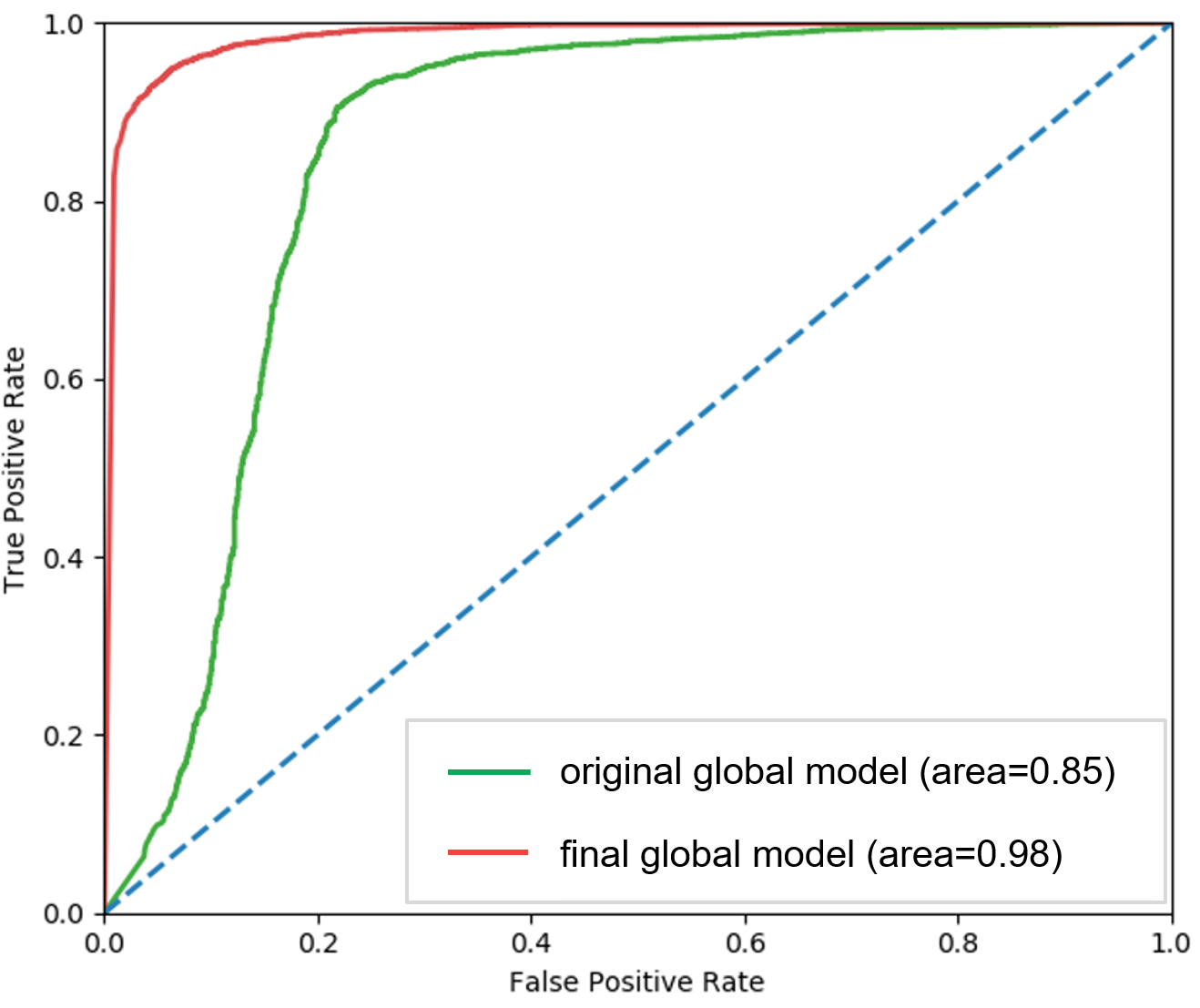}}
	\centering
	\vspace{-1.5mm}
	\caption{ROC of the original and final global model.}
	\label{fig-7}
\end{figure}

Besides, in order to validate the dropout-robustness of the proposed CFL protocol, we randomly set 15\% dropout clients in each training round. As shown in fig. \ref{fig-6}, in this case, the global model converges at 16 rounds, and the accuracy and loss value of the global model are 99.26\% and 0.0343, which are close to the values without dropout clients. Such a result shows that the proposed CFL protocol ensures the classification accuracy and the convergence rate of the global model facing dropout issues. Similar results can also be observed in Fig. \ref{fig-7}, where the final global model obtains AUC values of 0.97 and 0.98 on the Trec06p and Trec07 testing sets, which also reflect the dropout-robustness of the proposed CFL protocol.

More importantly, to demonstrate the communication efficiency of the proposed CFL protocol, we record the averaging communication times achieving a round of aggregation in TABLE I. Compared to the PPT protocol \cite{PPT2021}, CFL improves the communication efficiency by about 43.25\%. Such an improvement benefits from the cluster-based division. It is worth mentioning that the communication times are not necessarily negatively correlated with the number of dropout clients but are related to the distribution of network topology. 

\begin{table}[htbp] 
	\centering  
	\vspace{-10mm}
	\caption{Communication Times} 
	\begin{threeparttable}  
		\begin{tabular}{|p{5.5cm}<{\centering}|p{0.9cm}<{\centering}|p{0.9cm}<{\centering}|p{0.9cm}<{\centering}|p{0.9cm}<{\centering}|p{0.9cm}<{\centering}|p{0.9cm}<{\centering}|} 
			
			\hline  
			Percentage of dropout clients (\%) & 0 & 1 & 2 & 5 & 10 & 15 \\
			\hline  
			CFL's communication times & 105 & 106 & 104 & 105 & 102 & 92 \\ 
			\hline 
			PPT's communication times & 191 & 186 & 185 & 184 & 172 & 164  \\			
			\hline  					
		\end{tabular}
		
	\end{threeparttable}  
\end{table}   

In addition, we record the time of each operation in TABLE II, where the authenticated encryption (decryption) in our protocol consumes less time than the encryption (decryption) and signature (verification) in the PPT protocol \cite{PPT2021}, which achieves a 0.87\% improvement in terms of computational efficiency. TABLE II also shows the time for a single leader client in our protocol to perform noise generation, addition, and subtraction operations, which are almost the same as the PPT protocol.

\begin{table}[h] 
	\centering 
	\vspace{-10mm} 
	\caption{Computational Performance} 
	\begin{threeparttable}  
		\begin{tabular}{m{1.5cm}<{\centering}|m{6.5cm}<{\centering}|m{3.5cm}<{\centering}} 
			\hline  
			\hline  
			\makecell[c]{Index} & \makecell[c]{Operations$^{\dagger}$} & \makecell[c]{Time(ms/1000byte)}\\   
			\hline 
			\makecell[c]{1} & \makecell[c]{Encryption(AES-128bit)} & \makecell[c]{170.8248$ ^{\ast}$}\\  
			\hline  
			\makecell[c]{2} & \makecell[c]{Decryption(AES-128bit)} & \makecell[c]{0.0282$ ^{\ast} $}\\  
			\hline  
			\makecell[c]{3} & \makecell[c]{Signature(Elgamal-2048bit)} & \makecell[c]{0.0003$^{\ast}$}\\  
			\hline  
			\makecell[c]{4} & \makecell[c]{Verification(Elgamal-2048bit)} & \makecell[c]{0.0071$ ^{\ast} $}\\  
			\hline  
			\makecell[c]{5} & \makecell[c]{Encryption  (AES-GCM-128bit)} & \makecell[c]{169.3596}\\  
			\hline  
			\makecell[c]{6} & \makecell[c]{Decryption(AES-GCM-128bit)} & \makecell[c]{0.0163}\\
			\hline  
			\makecell[c]{7} & \makecell[c]{Noise generation } & \makecell[c]{1.2368}\\
			\hline  
			\makecell[c]{8} & \makecell[c]{Noise addition} & \makecell[c]{0.0877}\\
			\hline
			\makecell[c]{9} & \makecell[c]{Noise subtraction} & \makecell[c]{0.1316}\\
			\hline
			\hline 			
		\end{tabular}
		\begin{tablenotes}
			\footnotesize
			\item[$ \dagger $] Local\! model\! update\! parameters\! size: \!114MB\! in\! plaintext\! and\! 428.8MB\! in\! ciphertext.   
			\item[$ \ast $] Based on the same experimental settings, we directly use the results of relevant operation time in the PPT protocol \cite{PPT2021} for comparison.   
		\end{tablenotes}
	\end{threeparttable}  
\end{table}

\vspace{-10mm}
\section{Conclusion}

Aiming to solve the problems of a single point of failure and high communication costs in FL, we propose a CFL protocol for {\it FL in large-scale P2P networks}, focusing on efficient and privacy-preserving model training. Some interesting researches could further improve CFL in the future. On the one hand, an effective design for FL in dynamic P2P networks could flourish the development of decentralized FL in practice. On the other hand, along with the hardware development of clients, a computation-friendly homomorphic encryption algorithm could improve the communication efficiency of the system more.
%
%
%
\bibliographystyle{splncs04}
\bibliography{ref}
\end{document}